\documentclass[12pt]{article}
\usepackage{amsfonts}
\usepackage{amsmath}
\usepackage{amsthm}

\newtheorem{theorem}{Theorem}
\newtheorem{lemma}{Lemma}
\newtheorem{proposition}{Proposition}

\title{The exponential of the spin representation of the Lorentz algebra}
\author{jason hanson\thanks{\tt jhanson$@$digipen.edu}}

\begin{document}
\maketitle

\begin{abstract}
As discussed in a previous article, any (real) Lorentz algebra element possess a unique orthogonal decomposition as a sum of two mutually annihilating decomposable Lorentz algebra elements.  In this article, this concept is extended to the spin representation of the Lorentz algebra.  As an application, a formula for the exponential of the spin representation is obtained, as well as a formula for the spin representation of a proper orthochronous Lorentz transformation.
\end{abstract}

\section{Orthogonal decomposition}

Let $g$ be a Lorentz metric on ${\mathbb R}^4$.  That is, $g$ is a symmetric nondegenerate inner product with determinant $-1$.  For our purposes, we need not specify a signature for $g$.  The {\em Lorentz group} $O(g)$ is the Lie group of linear transformations $\Lambda$ on ${\mathbb R}^4$ such that $\Lambda^Tg\Lambda=g$, and the {\em Lorentz algebra} $so(g)$ is the Lie algebra of transformations $L$ for which $L^Tg+gL=0$.

\subsection{Decomposition of bivectors}

Elements of $so(g)$ are called {\em Lorentz bivectors.}  A special type of Lorentz bivector is the {\bf simple bivector}, which takes the form $u\wedge^gv$ for four--vectors $u,v$ in ${\mathbb R}^4$, where
\begin{equation}\label{eq:wedge}
  (u\wedge^gv)(w)\doteq g(v,w)u-g(u,w)v
\end{equation}
when applied to the four--vector $w$.  In index notation, $(u\wedge^gv)_\beta^\alpha=u^\alpha v_\beta-v^\alpha u_\beta$.  Simple bivectors are also called {\em decomposable bivectors,} and are characterized by the condition $\det(u\wedge^gv)=0$.

While not every Lorentz bivector $L$ is simple, it is the sum of simple bivectors.  In fact, it can be shown that any nonsimple Lorentz bivector $L$ admits an {\bf orthogonal decomposition}: $L=L_++L_-$, with $L_\pm$ simple and $L_+L_-=0=L_-L_+$.  This decomposition is unique.  Indeed, the summands of the decomposition of $L$ are given by
\begin{equation}\label{eq:orthdecomp}
  L_\pm=\pm\frac{L^3-\mu_\mp L}{\mu_+-\mu_-}
\end{equation}
where $\mu_\pm$ are the positive and negative roots of the equation $x^2+({\rm tr}_2L)x+\det{L}=0$, or equivalently, the solutions of the simultaneous equations
\begin{equation}\label{eq:mu}
  \mu_++\mu_-=-{\rm tr}_2L
  \quad\text{and}\quad
  \mu_+\mu_-=\det{L}.
\end{equation}
Here ${\rm tr}_2L$ is the second order trace of $L$, which can be computed by the formula ${\rm tr}_2L=-\tfrac{1}{2}{\rm tr}L^2$ (this identity holds for any traceless matrix).  In particular, ${\rm tr}_2L_\pm=-\mu_\pm$.  See \cite{me} for details.

\subsection{Spin representation of the Lorentz algebra}

Representations of $so(g)$ may be constructed from representations of the Clifford algebra ${\mathcal Cl}(g)$ on ${\mathbb R}^4$.  Recall that ${\mathcal Cl}(g)$ is the quotient of the tensor algebra $T^\ast({\mathbb R}^4)$ by the subalgebra generated by the relation $uv+vu=2g(u,v)$ for $u,v\in{\mathbb R}^4$.  Let $\rho$ be a Clifford algebra representation; i.e., a (possibly complex) vector space $V$ and a linear map $\rho:{\mathcal Cl}(g)\rightarrow{\rm Hom}(V,V)$ that respects Clifford multiplication: $\rho(uv)=\rho(u)\rho(v)$.  We obtain the {\em spin representation} $\sigma:so(g)\rightarrow{\rm Hom}(V,V)$ by setting
\begin{equation}\label{eq:spinrep}
  \sigma(u\wedge^gv)
  \doteq\tfrac{1}{4}\rho(uv-vu)
\end{equation}
for simple Lorentz bivectors, and extending linearly to all of $so(g)$.  One shows that $\sigma$ is a Lie algebra homomorphism: $\sigma([L_1,L_2])=\sigma(L_1)\sigma(L_2)-\sigma(L_2)\sigma(L_1)$ for all Lorentz bivectors $L_1,L_2$ (see \cite{Fulton-Harris}, for example).

A natural choice for the representation $\rho$ would be gamma matrices; i.e., $\rho(u)\doteq u^\alpha\gamma_\alpha$.  However, a representation may be constructed directly from the Clifford algebra itself: view $V={\mathcal Cl}(g)$ as a sixteen--dimensional real vector space, and take $\rho$ to be the identity.  Unlike the gamma matrix representation, this is not an irreducible Clifford algebra representation.  In the following, we will not have the need to make a specific choice for $\rho$, and we simply refer to $\sigma$ as ``the'' spin representation of $so(g)$.  We remark that all formulas, with the exception of those that appear in section \ref{sec:digression}, are actually valid for summands of the spin representation.  In particular, they are valid for half--spin representations.

They key property of the spin representation $\sigma$ that we will make use of is the following, which makes apparent the usefulness of decomposing a bivector into a sum of simple bivectors.  Here we write $I=\rho(1)$.

\begin{theorem}\label{thm:Lsq}
Suppose $L=u\wedge^gv$ is a simple Lorentz bivector.  Then ${\rm tr}_2L=g(u,u)g(v,v)-g(u,v)^2$ and $\sigma(L)^2=-\tfrac{1}{4}({\rm tr}_2L)I$.
\end{theorem}

\begin{proof}
From equation \eqref{eq:wedge}, one computes that ${\rm tr}_2L=-\tfrac{1}{2}{\rm tr}L^2$ is given by the stated expression (see also \cite{me}).  Now compute using equation \eqref{eq:spinrep} and the Clifford algebra relation:
\begin{align*}
  \sigma(u\wedge^gv)^2
  &=\tfrac{1}{16}\,\rho(uvuv-uv^2u-vu^2v+vuvu)\\
  &=\tfrac{1}{16}\,\rho\bigl(u[-uv+2g(v,u)]v-g(v,v)g(u,u)\\
  & \quad\quad\quad\quad
    -g(u,u)g(v,v)+v[-vu+2g(u,v)]u\bigr)\\
  &=\tfrac{1}{16}\,\rho\bigl(2g(u,v)(uv+vu)-4g(u,u)g(v,v)\bigr)
\end{align*}
which implies the stated expression for $\sigma(L)^2$.
\end{proof}

\subsection{Decomposition of a spin representation}

If $L$ is a Lorentz bivector, then $L^3$ is as well.  So we may apply the spin representation directly to each summand in equation \eqref{eq:orthdecomp} to obtain $\sigma(L_\pm)$ in terms of $\sigma(L)$ and $\sigma(L^3)$.  However, we would like an expression that involves only powers of $\sigma(L)$.

\begin{theorem}\label{thm:Lcube}
If $L=L_++L_-$ is the orthogonal decomposition of a nonsimple Lorentz bivector, then
$$\sigma(L_\pm)
  =\frac{\pm 2}{\mu_+-\mu_-}
   \left\{\tfrac{1}{4}(\mu_\mp+3\mu_\pm)\sigma(L)
          -\sigma(L)^3\right\}
$$
with $\mu_\pm$ as in equation \eqref{eq:mu}.
\end{theorem}

\begin{proof}
Since ($\ast$) $\sigma(L)=\sigma(L_+)+\sigma(L_-)$, we have that $\sigma(L)^3=\sigma(L_+)^3+3\sigma(L_+)^2\sigma(L_-)+3\sigma(L_+)\sigma(L_-)+\sigma(L_-)^3$.  Using theorem \ref{thm:Lsq} to reduce powers, we may rewrite this as ($\ast\ast$) $\sigma(L)^3=\tfrac{1}{4}(\mu_++3\mu_-)\sigma(L_+)+\tfrac{1}{4}(\mu_-+3\mu_+)\sigma(L_-)$.  The determinant of the linear system ($\ast$) and ($\ast\ast$) is $\tfrac{1}{2}(\mu_+-\mu_-)$, which is nonzero if $L$ is nonsimple, and the system may be solved to yield the stated expressions for $\sigma(L_\pm)$.
\end{proof}

The summands of the orthogonal decomposition of a nonsimple Lorentz bivector are mutually annihilating.  Although their images under the spin representation do not share this property, they do commute.

\begin{theorem}\label{thm:commute}
If $L=L_++L_-$ is the orthogonal decomposition of the nonsimple Lorentz bivector $L$, then $\sigma(L_+)\sigma(L_-)=\sigma(L_-)\sigma(L_+)=\tfrac{1}{8}({\rm tr}_2L)I+\tfrac{1}{2}\sigma(L)^2$.
\end{theorem}

\begin{proof}
As $L_+,L_-$ trivially commute, $[\sigma(L_+),\sigma(L_-)]=\sigma([L_+,L_-])=0$.  By theorem \ref{thm:Lsq} and equation \eqref{eq:mu}, we then have $\sigma(L)^2=\sigma(L_+)^2+2\sigma(L_+)\sigma(L_-)+\sigma(L_-)^2=2\sigma(L_+)\sigma(L_-)-\tfrac{1}{4}({\rm tr}_2L)I$.
\end{proof}

\subsection{A computational digression}\label{sec:digression}

To use the formula in theorem \ref{thm:Lcube}, we need to know the values of $\mu_\pm$, which are obtained from the invariants ${\rm tr}_2L$ and $\det{L}$ of $L$.  However, we would like to deduce these values in the event we only have knowledge of $\sigma(L)$.

\begin{lemma}
For any four--vectors $a,b,u,v$
$${\rm tr}\rho(abuv)
  ={\rm tr}I\left\{g(a,b)g(u,v)-g(a,u)g(b,v)
                   +g(a,v)g(b,v)\right\}.\qed
$$
\end{lemma}

\noindent
This generalizes the well--known identity for gamma matrices, so we need not repeat the computation here.  We do note however, that ${\rm tr}I={\rm tr}\rho(1)$ is the dimension of the representation: ${\rm tr}I=16$ for the Clifford algebra ${\mathcal Cl}(g)$ itself, and ${\rm tr}I=4$ for the gamma matrix representation.

\begin{lemma}
If $L=L_++L_-$ is the orthogonal decomposition of a Lorentz bivector with $L_+=a\wedge^gb$ and $L_-=u\wedge^gv$, then
$$g(a,v)g(b,u)-g(a,u)g(b,v)=0.$$
\end{lemma}

\begin{proof}
From equation \eqref{eq:wedge}, one computes $(a\wedge^gb)(u\wedge^gv)=g(b,u)av^Tg-g(b,v)au^Tg-g(a,u)bv^Tg+g(a,v)bu^Tg$.  Taking the trace of both sides, we get ${\rm tr}\{(a\wedge^gb)(u\wedge^gv)\}=2g(b,u)g(v,a)-2g(a,u)g(v,b)$, which is necessarily zero, since $(a\wedge^gb)(u\wedge^gv)=L_+L_-=0$.
\end{proof}

\begin{theorem}
If $L=L_++L_-$ is the orthogonal decomposition of a Lorentz bivector, then ${\rm tr}\{\sigma(L_+)\sigma(L_-)\}=0$.
\end{theorem}

\begin{proof}
Write $L_+=a\wedge^gb$ and $L_-=u\wedge^gv$.  Then $\sigma(L_+)\sigma(L_-)=\frac{1}{16}\rho(ab-ba)\rho(uv-vu)=\tfrac{1}{16}\bigl(\rho(abuv)-\rho(abvu)-\rho(bauv)+\rho(bavu)\bigr)$.  Take the trace and apply the previous two lemmas.
\end{proof}

\begin{theorem}
For any Lorentz bivector $L$, ${\rm tr}_2L=-4\,{\rm tr}\sigma(L)^2/{\rm tr}I$ and $\det{L}=4\,{\rm tr}\sigma(L)^4/{\rm tr}I-4\,{\rm tr}^2\sigma(L)^2/{\rm tr}^2I$.
\end{theorem}

\begin{proof}
For the first formula, take the trace of the formula in theorem \ref{thm:commute}.  For the second, use the orthogonal decomposition $L=L_++L_-$ and theorem \ref{thm:Lsq} to compute $\sigma(L)^4=\{\sigma(L_+)+\sigma(L_-)\}^4=\tfrac{1}{16}\mu_+^2I-\mu_+\sigma(L_+)\sigma(L_-)+\tfrac{3}{8}\mu_+\mu_-I-\mu_-\sigma(L_+)\sigma(L_-)+\tfrac{1}{16}\mu_-^2I$.  Taking traces, we get ${\rm tr}\sigma(L)^4=\tfrac{1}{16}({\rm tr}I)(\mu_+^2+6\mu_+\mu_-+\mu_-^2)=\tfrac{1}{16}({\rm tr}I)\{(\mu_++\mu_-)^2+4\mu_+\mu_-\}=\tfrac{1}{16}({\rm tr}I)({\rm tr}_2^2L+4\det{L})$, courtesy of equation \eqref{eq:mu}.  Solving for $\det{L}$ and using the first formula yields the desired expression.
\end{proof}

\section{Exponential of the spin representation}

By general principles, the exponential operation $\exp(L)\doteq\sum_{n\geq 0}L^n/n!$ is a map $\exp:so(g)\rightarrow O(g)$, whose image is the connected component $SO^+(g)$ of $O(g)$ containing the identity transformation; i.e., the set of all proper orthochronous Lorentz transformations.  A closed formula for $\exp(L)$ was obtained in \cite{Coll}.  Here we give a closed expression for $\exp(\sigma(L))$ for both simple and nonsimple Lorentz bivectors.

\begin{theorem}\label{thm:expsimple}
If $L$ is a simple Lorentz bivector, then $\exp(\sigma(L))=\bar{c}+\bar{s}\sigma(L)$, where
\begin{align*}
  \text{if ${\rm tr}_2L>0$,}\quad
  & \bar{c}\doteq\cos\tfrac{1}{2}\sqrt{{\rm tr}_2L}
  &\text{and}\quad
  & \bar{s}\doteq\frac{2}{\sqrt{{\rm tr}_2L}}\sin\tfrac{1}{2}\sqrt{{\rm tr}_2L}\\
  \text{if ${\rm tr}_2L<0$,}\quad
  & \bar{c}\doteq\cosh\tfrac{1}{2}\sqrt{-{\rm tr}_2L}
  &\text{and}\quad
  & \bar{s}\doteq\frac{2}{\sqrt{-{\rm tr}_2L}}\sinh\tfrac{1}{2}\sqrt{-{\rm tr}_2L}
\end{align*}
and if ${\rm tr}_2L=0$, then $\bar{c}=1=\bar{s}$.
\end{theorem}

\begin{proof}
In the case ${\rm tr}_2L>0$, theorem \ref{thm:Lsq} implies $\sigma(L)^{2p}=(-\theta^2)^p=(-1)^p\theta^{2p}$, where $\theta\doteq\tfrac{1}{2}\sqrt{{\rm tr}_2L}$.  Consequently, we may write $\sigma(L)^{2p+1}=\sigma(L)^{2p}\sigma(L)=(-1)^p\theta^{2p+1}\sigma(L)/\theta$.  Thus the series $\sum_{n\geq 0}\sigma(L)^n/n!=\sum_{p\geq 0}\sigma(L)^{2p}/(2p)!+\sum_{p\geq 0}\sigma(L)^{2p+1}/(2p+1)!$ is summed using the usual Taylor series expansion for sine and cosine.  The case when ${\rm tr}_2L<0$ is similar, and the case ${\rm tr}_2L=0$ is trivial.
\end{proof}

Recall that $\exp(A+B)=\exp(A)\exp(B)$ whenever the matrices $A,B$ commute.  Thus, $\exp(\sigma(L))=\exp(\sigma(L_+))\exp(\sigma(L_-))$.  Theorems \ref{thm:commute} and \ref{thm:expsimple} then lead to the following.

\begin{theorem}\label{thm:exp1}
Suppose $L=L_++L_-$ is the orthogonal decomposition of a nonsimple Lorentz bivector.  Define $\theta_\pm\doteq\tfrac{1}{2}\sqrt{\mp{\rm tr}_2L_\pm}$, $\bar{c}_+\doteq\cosh\theta_+$, $\bar{c}_-\doteq\cos\theta_-$, $\bar{s}_+\doteq\sinh\theta_+/\theta_+$, and $\bar{s}_-\doteq\sin\theta_-/\theta_-$.  Then
$$\exp(\sigma(L))
  =\bar{c}_+\bar{c}_-
   +\bar{s}_+\bar{c}_-\,\sigma(L_+)
   +\bar{c}_+\bar{s}_-\,\sigma(L_-)
   +\bar{s}_+\bar{s}_-\,\sigma(L_+)\sigma(L_-)\qed
$$
\end{theorem}

We derive an alternative formula for $\exp(\sigma(L))$ as a polynomial in $\sigma(L)$.  By combining theorems \ref{thm:Lcube} and \ref{thm:commute} with theorem \ref{thm:exp1}, we obtain the following.

\begin{theorem}\label{thm:exp2}
If $L$ is a nonsimple Lorentz bivector, then
$$\exp(\sigma(L))=\alpha_0+\alpha_1\sigma(L)+\alpha_2\sigma(L)^2
                  +\alpha_3\sigma(L)^3$$
\begin{align*}
  \alpha_0 &\doteq \bar{c}_+\bar{c}_-
           -\tfrac{1}{8}(\mu_++\mu_-)\bar{s}_+\bar{s}_-
  & \alpha_2 &\doteq \tfrac{1}{2}\bar{s}_+\bar{s}_-\\
\alpha_1 &\doteq\tfrac{1}{4}N\bigl\{(\mu_-+3\mu_+)\bar{s}_+\bar{c}_-
                                   -(\mu_++3\mu_-)\bar{c}_+\bar{s}_-\bigr\}
  & \alpha_3 &\doteq N(\bar{c}_+\bar{s}_--\bar{s}_+\bar{c}_-)
\end{align*}
with $\mu_\pm$ as in equation \eqref{eq:mu}, $\bar{c}_\pm,\bar{s}_\pm$ as in theorem \ref{thm:exp1}, and $N\doteq 2/(\mu_+-\mu_-)$.\qed
\end{theorem}

\section{Spin representation of a Lorentz transformation}

The spin representation $\sigma:so(g)\rightarrow{\rm Hom}(V,V)$ on the Lie algebra level induces a projective representation $\Sigma:SO^+(g)\rightarrow{\rm SL}(V)/\pm$ on the Lie group level (\cite{Wigner}).  We would like to deduce an explicit formula for $\Sigma$.

We define a Lorentz transformation $\Lambda\in SO^+(g)$ to be {\bf simple} if it is the image of a simple Lorentz bivector under the exponential map.  A criterion for simplicity is that ${\rm tr}_2\Lambda=2({\rm tr}\Lambda-1)$.  Here, the second order trace may be computed from the general formula ${\rm tr}_2\Lambda=\tfrac{1}{2}({\rm tr}^2\Lambda-{\rm tr}\Lambda^2)$.  Moreover, it should be noted that for any proper orthochronous Lorentz transformation (simple or not), ${\rm tr}\Lambda\geq 0$.  See \cite{me} for more details.

We will need the following fact for computing the logarithm of a simple Lorentz transformation, as given in \cite{me}.  The special case when ${\rm tr}\Lambda=0$ will be handled later.

\begin{proposition}\label{prop:simplelog}
If $\Lambda\in SO^+(g)$ is a simple Lorentz transformation with ${\rm tr}\Lambda>0$, then $\Lambda=\exp(L)$ and ${\rm tr}_2L=-\mu$, where $L=\tfrac{1}{2}k(\Lambda-\Lambda^{-1})$ and
\begin{enumerate}
  \item if $0<{\rm tr}\Lambda<4$, then $\displaystyle k=\frac{\sqrt{-\mu}}{\sin\sqrt{-\mu}}$ and $\sqrt{-\mu}=\cos^{-1}(\tfrac{1}{2}{\rm tr}\Lambda-1)$,
  \item if ${\rm tr}\Lambda>4$, then $\displaystyle k=\frac{\sqrt{\mu}}{\sinh\sqrt{\mu}}$ and $\sqrt\mu=\cosh^{-1}(\tfrac{1}{2}{\rm tr}\Lambda-1)$,
  \item if ${\rm tr}\Lambda=4$, then $k=0$ and $\mu=0$.
\end{enumerate}
\end{proposition}

\begin{theorem}\label{thm:simplespin}
Suppose $\Lambda\in SO^+(g)$ is simple.  If ${\rm tr}\Lambda>0$, then up to an overall sign,
$$\Sigma(\Lambda)
  =\frac{1}{2\sqrt{{\rm tr}\Lambda}}\left\{{\rm tr}\Lambda
   +2\sigma(\Lambda-\Lambda^{-1})\right\}.
$$
\end{theorem}

\begin{proof}
Let $L$ be a simple Lorentz bivector such that $\Lambda=\exp(L)$.  By the general properties of the exponential map on a Lie algebra, $\Sigma(\Lambda)=\exp(\sigma(L))$.  Writing $L=\tfrac{1}{2}k(\Lambda-\Lambda^{-1})$ as in proposition \ref{prop:simplelog}, we have $\exp(\sigma(L))=\bar{c}+\tfrac{1}{2}k\bar{s}\sigma(\Lambda-\Lambda^{-1})$, according to theorem \ref{thm:expsimple}.  The values of $\bar{c},\bar{s}$ depend on the value of ${\rm tr}_2L$.  We consider the case ${\rm tr}_2L>0$, so that $\mu<0$ in the notation of proposition \ref{prop:simplelog}, which occurs when $0<{\rm tr}\Lambda<4$.  The other two cases are similar.  In this case, we have $\cos\sqrt{-\mu}=\tfrac{1}{2}{\rm tr}\Lambda-1$.  Now, $\bar{c}=\cos\tfrac{1}{2}\sqrt{{\rm tr}_2L}=\cos\tfrac{1}{2}\sqrt{-\mu}$.  Similarly,
$$\tfrac{1}{2}k\bar{s}
  =\frac{1}{2}\frac{\sqrt{-\mu}}{\sin\sqrt{-\mu}}\frac{2}{\sqrt{-\mu}}
   \sin\tfrac{1}{2}\sqrt{-\mu}
  =\frac{\sin\tfrac{1}{2}\sqrt{-\mu}}{\sin\sqrt{-\mu}}
  =\frac{1}{2\cos\tfrac{1}{2}\sqrt{-\mu}}
$$
On the other hand, we have $\cos^2\tfrac{1}{2}\sqrt{-\mu}=\tfrac{1}{2}(1-\cos\sqrt{-\mu})=\tfrac{1}{4}{\rm tr}\Lambda$, so that $\cos\tfrac{1}{2}\sqrt{-\mu}=\pm\tfrac{1}{2}\sqrt{{\rm tr}\Lambda}$.  Although the choice of $L$ determines the sign here, the choice of $L$ such that $\exp(L)=\Lambda$ is not unique.  Indeed, if we take $L'=\alpha L$, with $\alpha$ chosen such that $\sqrt{{\rm tr}_2L'}=\sqrt{{\rm tr}_2L}+2\pi$, then $\exp(L')=\Lambda$.  However, $\tfrac{1}{2}\sqrt{-\mu'}=\tfrac{1}{2}\sqrt{-\mu}+\pi$, so that $\exp(\sigma(L'))=-\exp(\sigma(L))$.
\end{proof}

To obtain an analogous formula for a nonsimple Lorentz transformation, we will make use of the fact that such a transformation is a product of commuting simple transformations.  Indeed, since $SO^+(g)$ is exponential, we may write $\Lambda=\exp(L)$ for some Lorentz bivector $L$.  Using the orthogonal decomposition $L=L_++L_-$, we have $\exp(L)=\exp(L_+)\exp(L_-)$.  Taking $\Lambda_\pm\doteq\exp(L_\pm)$, we may write $\Lambda=\Lambda_+\Lambda_-$, with $\Lambda_+,\Lambda_-$ commuting simple Lorentz transformations.  In \cite{me}, the following explicit formula is obtained.

\begin{proposition}\label{prop:factor}
If $\Lambda\in SO^+(g)$ is nonsimple, then $\Lambda=\Lambda_+\Lambda_-$, where $\Lambda_\pm$ are the commuting simple Lorentz transformations
$$\Lambda_\pm
  =\pm\frac{1}{2(c_+-c_-)}\bigl\{
     (1+2c_\pm)I-\Lambda^{-1}-(1+2c_\mp)\Lambda+\Lambda^2\bigr\}
$$
with $c_\pm=\tfrac{1}{4}({\rm tr}\Lambda\pm\sqrt\Delta)$ and $\Delta\doteq{\rm tr}^2\Lambda-4{\rm tr}_2\Lambda+8$, $c_+>1$, and $1\leq c_-<1$.
\end{proposition}

Using this decomposition, we use the fact that $\Sigma$ is a group homomorphism to write $\Sigma(\Lambda)=\Sigma(\Lambda_+)\Sigma(\Lambda_-)$, where each factor on the right hand side may be computed using theorem \ref{thm:simplespin}.  Note that $\Sigma(\Lambda_\pm)$ necessarily commute.

\begin{theorem}
If $\Lambda$ is a nonsimple proper orthochronous Lorentz transformation with $2+2{\rm tr}\Lambda+{\rm tr}_2\Lambda\neq 0$, then up to sign
\begin{multline*}
  \Sigma(\Lambda)
  =\frac{1}{2\sqrt{2+2{\rm tr}\Lambda+{\rm tr}_2\Lambda}}
    \bigl\{(2+{\rm tr}\Lambda+{\rm tr}_2\Lambda
                       -\tfrac{1}{4}{\rm tr}^2\Lambda)\\
    +({\rm tr}\Lambda+2)\sigma(\Lambda-\Lambda^{-1})
              -\sigma(\Lambda^2-\Lambda^{-2})
              +\sigma(\Lambda-\Lambda^{-1})^2\bigr\}
\end{multline*}
\end{theorem}

\begin{proof}
For brevity, we set $\tau_k^\pm\doteq{\rm tr}_k\Lambda_\pm$.  From theorem \ref{thm:simplespin}, we compute $\Sigma(\Lambda)=\Sigma(\Lambda_+)\Sigma(\Lambda_-)$ to be:
\begin{multline}\label{eq:prod}
  \Sigma(\Lambda)
  =\frac{1}{4\sqrt{\tau_1^+\tau_1^-}}
   \bigl\{\tau_1^+\tau_1^-
          +2\tau_1^+\sigma(\Lambda_--\Lambda_-^{-1})
          +2\tau_1^-\sigma(\Lambda_+-\Lambda_+^{-1})\\
          +4\sigma(\Lambda_+-\Lambda_+^{-1})\sigma(\Lambda_--\Lambda_-^{-1})
   \bigr\}
\end{multline}
Now from proposition \ref{prop:factor}, one shows that ($\star$) $\tau_1\doteq{\rm tr}\Lambda=2(c_++c_-)$ and $\tau_2\doteq{\rm tr}_2\Lambda=4c_+c_-+2$, and that ($\star\star$) $\tau_1^\pm=2(1+c_\pm)$.  Moreover, since for any Lorentz transformation $\Lambda^{-1}=g^{-1}\Lambda^Tg$, we find that
\begin{equation}\label{eq:dLambdapm}
  \Lambda_\pm-\Lambda_\pm^{-1}=\mp\frac{1}{2(c_+-c_-)}
                                  \bigl\{2c_\mp(\Lambda-\Lambda^{-1})
                                         -(\Lambda^2-\Lambda^{-2})\bigr\}
\end{equation}
We now rewrite the summands of equation \eqref{eq:prod} in terms of $\Lambda$ and its first and second order traces.  For the first summand, using ($\star\star$) and $(\star)$ we obtain
\begin{equation}\label{eq:1summand}
  \tau_1^+\tau_1^-=4(1+c_+)(1+c_-)
  =2+2\tau_1+\tau_2
\end{equation}
For the second and third summands in \eqref{eq:prod}, using ($\star\star$) and \eqref{eq:dLambdapm} one computes $\tau_1^+\sigma(\Lambda_--\Lambda_-^{-1})+\tau_1^-\sigma(\Lambda_+-\Lambda_+^{-1})=2(1+c_++c_-)\sigma(\Lambda-\Lambda^{-1})-\sigma(\Lambda^2-\Lambda^{-2})$.  Thus by ($\star$),
\begin{equation}\label{eq:23summand}
  \tau_1^+\sigma(\Lambda_--\Lambda_-^{-1})
    +\tau_1^-\sigma(\Lambda_+-\Lambda_+^{-1})
  =(\tau_1+2)\sigma(\Lambda-\Lambda^{-1})
    -\sigma(\Lambda^2-\Lambda^{-2})
\end{equation}
For the fourth summand in \eqref{eq:prod}, we similarly compute (although we also need to use the fact that $\sigma(\Lambda-\Lambda^{-1})$ and $\sigma(\Lambda^2-\Lambda^{-2})$ commute, as $\sigma$ is a Lie algebra homomorphism)
\begin{multline}\label{eq:4summand1}
  4\sigma(\Lambda_+-\Lambda_+^{-1})\sigma(\Lambda_--\Lambda_-^{-1})
  =\frac{1}{(c_+-c_-)^2}
    \bigl\{(2-\tau_2)\sigma(\Lambda-\Lambda^{-1})^2\\
           +\tau_1\sigma(\Lambda-\Lambda^{-1})\sigma(\Lambda^2-\Lambda^{-2})
           -\sigma(\Lambda^2-\Lambda^{-2})^2\bigr\}
\end{multline}
However, the terms in this expression are not algebraically independent.  To find a relation, we make use of the fact that $\Lambda_\pm$ are simple: from theorem \ref{thm:Lsq}, $\sigma(L_\pm)^2=-\tfrac{1}{4}{\rm tr}_2L_\pm$, where $L_\pm$ are such that $\exp(L_\pm)=\Lambda_\pm$, and may be computed using proposition \ref{prop:simplelog}.  Indeed, $L_+=\tfrac{1}{2}k_+(\Lambda_+-\Lambda_+^{-1})$, with ${\rm tr}_2L_+=-\mu_+$, $c_+=\cosh\sqrt{\mu_+}$, and $k_+=\sqrt{\mu_+}/\sinh\sqrt{\mu_+}$.  The equation $\sigma(L_+)^2=-\tfrac{1}{4}{\rm tr}_2L_+$ and \eqref{eq:dLambdapm} then lead to
\begin{multline*}
  4c_-^2\sigma(\Lambda-\Lambda^{-1})^2
  -4c_-\sigma(\Lambda-\Lambda^{-1})\sigma(\Lambda^2-\Lambda^{-2})
  +\sigma(\Lambda^2-\Lambda^{-2})^2\\
  =4(c_+-c_-)^2(c_+^2-1)
\end{multline*}
(note that $\sinh^2\sqrt{\mu_+}=c_+^2-1$).  Similarly, $\sigma(L_-)^2=-\tfrac{1}{4}{\rm tr}_2L_-$, proposition \ref{prop:simplelog}, and \eqref{eq:dLambdapm} imply
\begin{multline*}
  4c_+^2\sigma(\Lambda-\Lambda^{-1})^2
  -4c_+\sigma(\Lambda-\Lambda^{-1})\sigma(\Lambda^2-\Lambda^{-2})
  +\sigma(\Lambda^2-\Lambda^{-2})^2\\
  =4(c_+-c_-)^2(c_-^2-1)
\end{multline*}
Adding these two equations together and using ($\star$) then yields the relation
\begin{multline}\label{eq:4summand2}
  (\tau_1^2-2\tau_2+4)\sigma(\Lambda-\Lambda^{-1})^2
  -2\tau_1\sigma(\Lambda-\Lambda^{-1})\sigma(\Lambda^2-\Lambda^{-2})\\
  +2\sigma(\Lambda^2-\Lambda^{-2})^2
  =(c_+-c_-)^2(\tau_1^2-2\tau_2-4)
\end{multline}
Combining \eqref{eq:4summand1} and \eqref{eq:4summand2} together, we see that we may write the fourth summand in \eqref{eq:prod} as
\begin{equation}\label{eq:4summand}
  4\sigma(\Lambda_+-\Lambda_+^{-1})\sigma(\Lambda_--\Lambda_-^{-1})
  =2\sigma(\Lambda-\Lambda^{-1})^2
   -\tfrac{1}{2}(\tau_1^2-2\tau_2-4)
\end{equation}
Combining \eqref{eq:prod}, \eqref{eq:1summand}, \eqref{eq:23summand}, and \eqref{eq:4summand} give the desired formula for $\Sigma(\Lambda)$.
\end{proof}

\subsection{Special case}

We consider the case when a Lorentz transformation $\Lambda\in SO^+(g)$ satisfies the identity
\begin{equation}\label{eq:special}
  2+2{\rm tr}\Lambda+{\rm tr}_2\Lambda=0
\end{equation}
In the case when $\Lambda$ is simple, so that ${\rm tr}_2\Lambda=2({\rm tr}\Lambda-1)$, this condition reduces to ${\rm tr}\Lambda=0$.  It should be noted that a simple Lorentz transformation is traceless only if $\Lambda=\exp(L)$ for some simple Lorentz bivector with ${\rm tr}_2L=\pi^2$.  A nonsimple transformation satisfies \eqref{eq:special} only if its decomposition in proposition \ref{prop:factor}, $\Lambda=\Lambda_+\Lambda_-$, is such that the simple factor $\Lambda_-$ is traceless.

We will not be able to obtain an explicit formula for the spin representation $\Sigma(\Lambda)$ of such a Lorentz transformation.  However, we can give an algorithm.  The key lies in the following two facts from \cite{me}.

\begin{proposition}
Let $u,v$ be four--vectors, and set $L\doteq u\wedge^gv$.  The two--plane ${\mathcal P}$ spanned by $u,v$ is nondegenerate (that is, $g$ is nondegenerate when restricted to ${\mathcal P}$) if and only if ${\rm tr}_2L\neq 0$, in which case $P_L\doteq -L^2/{\rm tr}_2L$ is $g$--orthogonal projection onto ${\mathcal P}$.  Conversely, if $P$ is $g$--orthogonal projection onto a two--plane, then the two--plane is nondegenerate and $P=P_L$, where $L=u\wedge^gv$ and $u,v$ are any linearly independent four--vectors in the image of $P$.
\end{proposition}

\begin{proposition}
If $\Lambda$ is a simple Lorentz transformation with ${\rm tr}\Lambda=0$, then $\Lambda^2=I$, $P_\Lambda\doteq\tfrac{1}{2}(I-\Lambda)$ is $g$--orthogonal projection onto a nondegenerate two--plane, and $-\pi^2P_\Lambda$ is the square of a simple Lorentz bivector $L_\Lambda$ with $\exp(L_\Lambda)=\Lambda$ and ${\rm tr}_2L_\Lambda=\pi^2$.
\end{proposition}

\begin{theorem}
Suppose $\Lambda\in SO^+(g)$ is simple with ${\rm tr}\Lambda=0$.  Let $u,v$ be any two linearly independent vectors in the image of $P_\Lambda=\tfrac{1}{2}(I-\Lambda)$.  Then up to sign, $\Sigma(\Lambda)=(2/\sqrt{{\rm tr}_2(u\wedge^gv)})\sigma(u\wedge^gv)$.
\end{theorem}

\begin{proof}
Set $L\doteq u\wedge^gv$.  Then $P_L$ is $g$--orthogonal projection onto the (necessarily nondegenerate) two--plane ${\mathcal P}$ spanned by $u,v$.  By uniqueness of $g$--orthogonal projection, we must have $P_L=P_\Lambda$.  Now $L_\Lambda=u'\wedge^gv'$, where $u',v'$ lie in ${\mathcal P}$.  Writing $u'=au+bv$ and $v'=cu+dv$, we compute that $u'\wedge^gv'=(ad-bc)u\wedge^gv$; i.e., $L_\Lambda=\alpha L$ for some scalar $\alpha$.  Taking 2--traces, we get $\pi^2={\rm tr}_2L_\Lambda=\alpha^2{\rm tr}_2L$, so that $\alpha=\pm\pi/\sqrt{{\rm tr}_2L}$.  Since $\exp(L_\Lambda)=\Lambda$, $\Sigma(\Lambda)=\exp\sigma(L_\Lambda)$.  On the other hand, by theorem \ref{thm:expsimple}, $\exp\sigma(L_\Lambda)=(2/\pi)\sigma(L_\Lambda)=(2\alpha/\pi)\sigma(L)$.
\end{proof}


\end{document}